\pdfoutput=1

\documentclass[letterpaper, 10 pt, conference]{ieeeconf}  

\IEEEoverridecommandlockouts                              
\usepackage{graphics} 
\usepackage{amsmath} 
\usepackage{amssymb}  

\usepackage{cite}
\usepackage{graphicx}
\usepackage{pgf}
\usepackage{color}
\usepackage[textsize=tiny,textwidth=2cm]{todonotes}	

\newcommand{\be}{\begin{enumerate}}
\newcommand{\ee}{\end{enumerate}}

\newcommand{\D}{\mathsf D}
\newcommand{\R}{\mathsf R}
\newcommand{\Adj}{\text{Adj}}

\newcommand{\Prob}{\mathbb P}

\newtheorem{thm}{Theorem}
\newtheorem{lem}[thm]{Lemma}
\newtheorem{definition}{Definition}
\newtheorem{rem}{Remark}

\title{\LARGE \bf
On Differentially Private Filtering for Event Streams
}

\author{Jerome Le Ny
\thanks{J. Le Ny is with the department of Electrical Engineering, Ecole Polytechnique de Montreal,
QC H3T-1J4, Canada.
        {\tt\small jerome.le-ny@polymtl.ca}}%
}

\begin{document}

\maketitle
\thispagestyle{empty}
\pagestyle{empty}





\begin{abstract}

Rigorous privacy mechanisms that can cope with dynamic data are required to encourage
a wider adoption of large-scale monitoring and decision systems relying on end-user information.
A promising approach to develop these mechanisms is to specify quantitative privacy requirements
at design time rather than as an afterthought, and to rely on signal processing techniques to
achieve satisfying trade-offs between privacy and performance specifications.
This paper discusses, from the signal processing point of view, an event stream analysis problem 
introduced in the database and cryptography literature. A discrete-valued input 
signal describes the occurrence of events contributed by end-users, and a system is supposed 
to provide some output signal based on this information, while preserving the privacy 
of the participants. The notion of privacy adopted here is that of event-level differential privacy, 
which provides strong privacy guarantees and has important operational advantages. Several 
mechanisms are described to provide differentially private output signals while minimizing the
impact on performance. These mechanisms demonstrate the benefits of leveraging system theoretic
techniques to provide privacy guarantees for dynamic systems.

\end{abstract}

\section{Introduction}

Privacy issues associated with emerging large-scale monitoring and decision systems are receiving
an increasing amount of attention. Indeed, privacy concerns are already resulting in delays or cancellations
in the deployment of smart grids, location-based services, or civilian unmanned aerial systems \cite{EPIC03_privacyCenter}.
In order to encourage the adoption of these systems, which can have important societal benefits, 
new mechanisms providing clear and rigorous privacy protection guarantees are needed.

Unfortunately, providing such guarantees for a system generally involves sacrificing some level of performance. 
Evaluating the resulting trade-offs rigorously requires a quantitative definition of privacy, and in the last few years
the notion of differential privacy has emerged essentially as a standard specification \cite{Dwork_ICAL06_DP}.
Intuitively, a system receiving inputs from end-users is differentially private if one cannot infer from its observable 
behavior if any specific individual contributed its data or not. Other quantitative notions of privacy have been proposed, 
e.g., \cite{Duncan86_disclosure, Sankar11_privacyInfoTheoretic}, but the differential privacy definition has important 
operational advantages. In particular, it does not require modeling the available auxiliary information that can be 
linked to the output of the system of interest to create privacy breaches. Moreover, it is an achievable privacy goal 
despite the fact that a database on which an individual has no influence could still potentially leak information 
about her in the presence of arbitrary auxiliary information \cite{Dwork_ICAL06_DP}.

Nevertheless, differential privacy is a very strong notion of privacy and might require large perturbations to the
published results of an analysis in order to hide the presence of individuals. This is especially true for applications 
where users continuously contribute data over time, and it is thus important to design advanced mechanisms that 
can limit the impact on system performance of differential privacy requirements. Previous work on designing differentially
private mechanisms for the publication of time-series include \cite{Rastogi10_DPtimeSeries, Li11_DPcompressive},
but these mechanisms are not causal and hence not suited for real-time applications. The papers 
\cite{Dwork10_DPcounter, Chan11_DPcontinuous, Bolot11_DPdecayingSums}
provide real-time mechanisms to approximate a few specific filters transforming user-contributed input event streams into
public output streams. For example, \cite{Dwork10_DPcounter, Chan11_DPcontinuous} consider a private accumulator
providing the total number of events that occured in the past. This paper is inspired by this scenario, and builds
on our previous work on this problem \cite[Section IV]{LeNy_DP_CDC12} \cite[Section VI]{LeNyDP2012_journalVersion}.

The rest of the paper is organized as follows.
Section \ref{section: differential privacy background} provides some technical background on differential privacy
and describes a basic mechanism enforcing privacy by injecting white Gaussian noise. 
Section \ref{section: problem presentation} describes the real-time event stream filtering scenario of interest.
In Section \ref{eq: linear equalization mechanisms}, we optimize architectures based on linear estimators to 
provide real-time private filters with reduced impact on performance. Section \ref{eq: DF mechanisms} attempts at 
leveraging the knowledge that the input stream takes values in a discrete set, by considering slightly non-linear 
structures based on decision-feedback equalization. Finally, we conclude with a brief illustrative example
in Section \ref{eq: example}.




\section{Differential Privacy}	\label{section: differential privacy background}

In this section we review the notion of differential privacy \cite{Dwork06_DPcalibration} 
as well as a basic mechanism that can be used to achieve it when the released data 
belongs to a finite-dimensional vector space. We refer the reader to the surveys by Dwork, 
e.g., \cite{Dwork_ICAL06_DP}, for additional background on differential privacy, and 
to \cite{LeNyDP2012_journalVersion} for the proofs of the results in this section.

\subsection{Definition}

Let us fix some probability space $(\Omega, \mathcal F, \Prob)$. Let $\D$ be a space of datasets of interest 
(e.g., a space of data tables, or a signal space). A \emph{mechanism} is a map $M: \D \times \Omega \to \R$, 
for some measurable output space $\R$, such that for any element $d \in \D$, $M(d,\cdot)$ is a random variable, 
typically writen simply $M(d)$. A mechanism can be viewed as a probabilistic algorithm to answer a query $q$, 
which is a map $q: \D \to \R$. 

%
Next, we introduce the definition of differential privacy.  
Intuitively in the following definition, $\D$ is a space of datasets of interest, and we have a symmetric binary relation 
$\Adj$ on $\D$, called adjacency, such that $\Adj(d,d')$ if and only if $d$ and $d'$ differ by the data of a single participant.

\begin{definition}	\label{def: differential privacy original}
Let $\D$ be a space equipped with a symmetric binary relation denoted $\Adj$, and let $(\R, \mathcal M)$ 
be a measurable space. Let $\epsilon, \delta \geq 0$. A mechanism $M: \D \times \Omega \to \R$ is 
$(\epsilon, \delta)$-differentially private if for all $d,d' \in \D$ such that $\Adj(d,d')$, we have
\begin{align}	\label{eq: standard def approximate DP original}
\Prob(M(d) \in S) \leq e^{\epsilon} \Prob(M(d') \in S) + \delta, \;\; \forall S \in \mathcal M. 
\end{align}
If $\delta=0$, the mechanism is said to be $\epsilon$-differentially private. 
\end{definition}

The definition says that for two adjacent datasets, the distributions over the outputs of the mechanism should be close.  
The choice of the parameters $\epsilon, \delta$ is set by the privacy policy. Typically $\epsilon$ is taken to be a small 
constant, e.g., $\epsilon \approx 0.5$ or perhaps even $\ln p$ for some small $p \in \mathbb N$. The parameter $\delta$ 
should be kept small as it controls the probability of certain significant losses of privacy, e.g., when a zero probability 
event for input $d'$ becomes an event with positive probability for input $d$ in (\ref{eq: standard def approximate DP original}).

A fundamental property of the notion of differential privacy is that no additional privacy loss can occur by simply manipulating an output 
that is differentially private. 
To state it, recall that a probability kernel between two measurable spaces $(\R_1,\mathcal M_1)$ and $(\R_2,\mathcal M_2)$ is
a function $k: \R_1 \times \mathcal M_2 \to [0,1]$ such that $k(\cdot,S)$ is measurable for each $S \in \mathcal M_2$ and
$k(r,\cdot)$ is a probability measure for each $r \in \R_1$.

\begin{thm}[Resilience to post-processing]	\label{thm: resilience to post-processing}
Let $M_1: \D \times \Omega \to (\R_1,\mathcal M_1)$ be an $(\epsilon,\delta)$-differentially private mechanism.
Let $M_2: \D \times \Omega \to (\R_2,\mathcal M_2)$ be another mechanism, such that there exists
a probability kernel $k: \R_1 \times \mathcal M_2 \to [0,1]$ verifying
\begin{align}	\label{eq: post-processing definition}
\Prob(M_2(d) \in S | M_1(d)) = k(M_1(d),S), \; \text{a.s.}, 
\end{align}
for all $S \in \mathcal M_2$ and $d \in \D$. Then $M_2$ is $(\epsilon,\delta)$-differentially private. 
\end{thm}

Note that in (\ref{eq: post-processing definition}), the kernel $k$ is not allowed to depend on the dataset $d$. In other words, this condition 
says that once $M_1(d)$ is known, the distribution of $M_2(d)$ does not further depend on $d$. The theorem says that a mechanism $M_2$ 
accessing a dataset only indirectly via the output of a differentially private mechanism $M_1$ cannot weaken the privacy guarantee.


\subsection{A Basic Differentially Private Mechanism}	\label{section: basic mech}

A mechanism that throws away all the information in a dataset is obviously private, 
but not useful, and in general one has to trade off privacy for utility when 
answering specific queries. 
We recall below a basic mechanism that can be used to
answer queries in a differentially private way.
We are only concerned in this section with queries that return numerical answers, 
i.e., here a query is a map $q: \D \to \R$, where the output space $\R$ equals $\mathbb R^k$ for some $k > 0$,
is equipped with a norm  denoted $\| \cdot \|_\R$,
and the $\sigma$-algebra $\mathcal M$ on $\R$ is taken to be the standard 
Borel $\sigma$-algebra. 
The following quantity plays an important
role in the design of differentially private mechanisms \cite{Dwork06_DPcalibration}.

\begin{definition}	\label{defn: sensitivity}
Let $\D$ be a space equipped with an adjacency relation $\Adj$.
The sensitivity of a query $q: \D \to \R$ is defined as
\[
\Delta_\R q := \max_{d,d':\Adj(d,d')} \|q(d) - q(d')\|_\R.
\]
In particular, for $\R = \mathbb R^k$ equipped with the $p$-norm 
$\| x \|_p = \left(\sum_{i=1}^k |x_i|^p \right)^{1/p}$, 
for $p \in [1,\infty]$,
we denote the $\ell_p$ sensitivity by $\Delta_p q$.
\end{definition}

A differentially private mechanism proposed in  \cite{Dwork06_DPgaussian} modifies an answer to a 
numerical query by adding iid zero-mean Gaussian noise.
Recall the definition of the $\mathcal Q$-function 
\[
\mathcal Q(x) := \frac{1}{\sqrt{2 \pi}} \int_x^{\infty} e^{-\frac{u^2}{2}} du.
\]
We have the following theorem \cite{Dwork06_DPgaussian, LeNyDP2012_journalVersion}.

\vspace{0.1cm}
\begin{thm}	\label{thm: Gaussian mech}
Let $q: \D \to \mathbb R^k$ be a query.
Then the Gaussian mechanism $M_q: \D \times \Omega \to \mathbb R^k$ 
defined by $M_q(d) = q(d) + w$, with $w \sim \mathcal N\left(0,\sigma^2 I_k \right)$, 
where $\sigma \geq \frac{\Delta_2 q}{2 \epsilon}(K + \sqrt{K^2+2\epsilon})$ and $K = \mathcal Q^{-1}(\delta)$,
is $(\epsilon,\delta)$-differentially private.
\end{thm}

\vspace{0.1cm}
For the rest of the paper, we define 
\[
\kappa_{\delta,\epsilon} = \frac{1}{2 \epsilon} (K+\sqrt{K^2+2\epsilon}),
\]
so that the standard deviation $\sigma$ in Theorem \ref{thm: Gaussian mech}
can be written $\sigma(\delta,\epsilon) = \kappa_{\delta,\epsilon} \Delta_2 q$.
It can be shown that $\kappa_{\delta,\epsilon}$ behaves roughly as $O(\ln(1/\delta))^{1/2}/\epsilon$.
For example, to guarantee $(\epsilon,\delta)$-differential privacy with $\epsilon = \ln(2)$ and $\delta = 0.05$, 
the standard deviation of the Gaussian noise introduced should be about $2.65$ times the $\ell_2$-sensitivity of $q$.


\section{Filtering Event Streams}	\label{section: problem presentation}

We now turn to the description of our scenario of interest, similar to the one introduced 
in \cite{Dwork10_DPcounter, Chan10_counter}. A system receives an input signal $u = \{u_t\}_{t \geq 0}$ with 
values in the discrete set $\left\{ \pm \frac{k}{2}, k \in \mathbb N  \right \}$. Such a signal can for example record 
the number of occurrences of certain events of interest at each period (we centered the values around zero for 
convenience later on). Similarly to \cite{Dwork10_DPcounter, Chan10_counter}, two signals $u$ and $u'$ are 
adjacent if and only if they differ at a single time by at most $d$, or equivalently
\begin{equation}	\label{eq: adjacency event-level DP}
\Adj^d(u,u') \text{ iff } u-u' = k \; \delta_{t_0}, |k| \leq d,\text{for some } t_0, 
\end{equation}
where $\delta_{t_0}$ denotes the discrete impulse at $t_0$.
The motivation for this adjacency relation is that a given individual contributes events to the stream at a single 
time only, and we want to preserve \emph{event-level privacy} \cite{Dwork10_DPcounter}, that is, hide to some 
extent the presence or absence of an event at a particular time. This could for example prevent the inference of 
individual transactions from publicly available collaborative filtering outputs, as in \cite{Calandrino11_privacyAttackCollabFilt}.

Even though individual events should be hidden, we would like to release a filtered version $Fu$ of the original signal,
where $F$ is a given causal stable linear time-invariant system. Note that in this paper, all signals and filter coefficients
are assumed to be real-valued, and all systems are single-input single-output.
%
%
%
Privacy preserving approximations of $F$ can be developed based on the following sensitivity calculation.
\begin{lem}
Let $G$ be a linear time-invariant system with impulse response $g := \{g_t\}_t$. Then, for the adjacency 
relation (\ref{eq: adjacency event-level DP}) on binary-valued input signals, the $\ell_p$ sensitivity of $G$ is
$\Delta_p G = d \|g\|_p$. In particular for $p=2$, we have $\Delta_2 G = d \|G\|_2$, where $\|G\|_2$ is 
the $\mathcal H_2$ norm of $G$.
\end{lem}

\begin{proof}
For two adjacent binary-valued signals $u, u'$, we have
\begin{align*}
\|Gu - Gu'\|_p &= \|G(u-u')\|_p = d \|g * \delta_{t_0}\|_p \\
&= d \|\{g_{t-t_0}\}_{t}\|_p = d \| g \|_p.
\end{align*}
\end{proof}

This leads to the following theorem, generalizing Theorem \ref{thm: Gaussian mech} to dynamic systems. 
Certain technical measurability issues in the proof of this result are resolved in \cite{LeNyDP2012_journalVersion}.
\begin{thm}	\label{eq: basic DP mechanism}
The mechanism $M(u) = Gu + n$, where $n$ is a Gaussian white noise with covariance $d^2 \kappa_{\delta,\epsilon}^2 \|G\|_2^2$,
is $(\epsilon,\delta)$-differentially private for the adjacency relation (\ref{eq: adjacency event-level DP}).
\end{thm}

Theorem \ref{eq: basic DP mechanism} can now be combined with Theorem \ref{thm: resilience to post-processing} to obtain a
family of privacy preserving mechanisms approximating $F$, as illustrated on Fig. \ref{fig: DP filter approximation}. On that figure, the signal $v$ is
differentially private, and hence $\hat y$ as well by the resilience to post-processing property (Theorem \ref{thm: resilience to post-processing}). 
Two extreme cases include $G=\text{id}$, called input perturbation, and $H=\text{id}$, called output perturbation. In general however, 
these two choices can exhibit very poor performance \cite{LeNy_DP_CDC12}. Throughout this paper, we measure 
the precision of specific approximations  by the mean square error (MSE) between the published and desired outputs, i.e., 
\[
\lim_{T \to \infty} \frac{1}{T} \sum_{t=0}^\infty \mathbb E[|e_t|^2],
\]
with $e = y - \hat y$. The next section is devoted to the description of two ways of choosing the filters $G, H$ as linear filters.

\begin{figure}
\includegraphics[width=\linewidth]{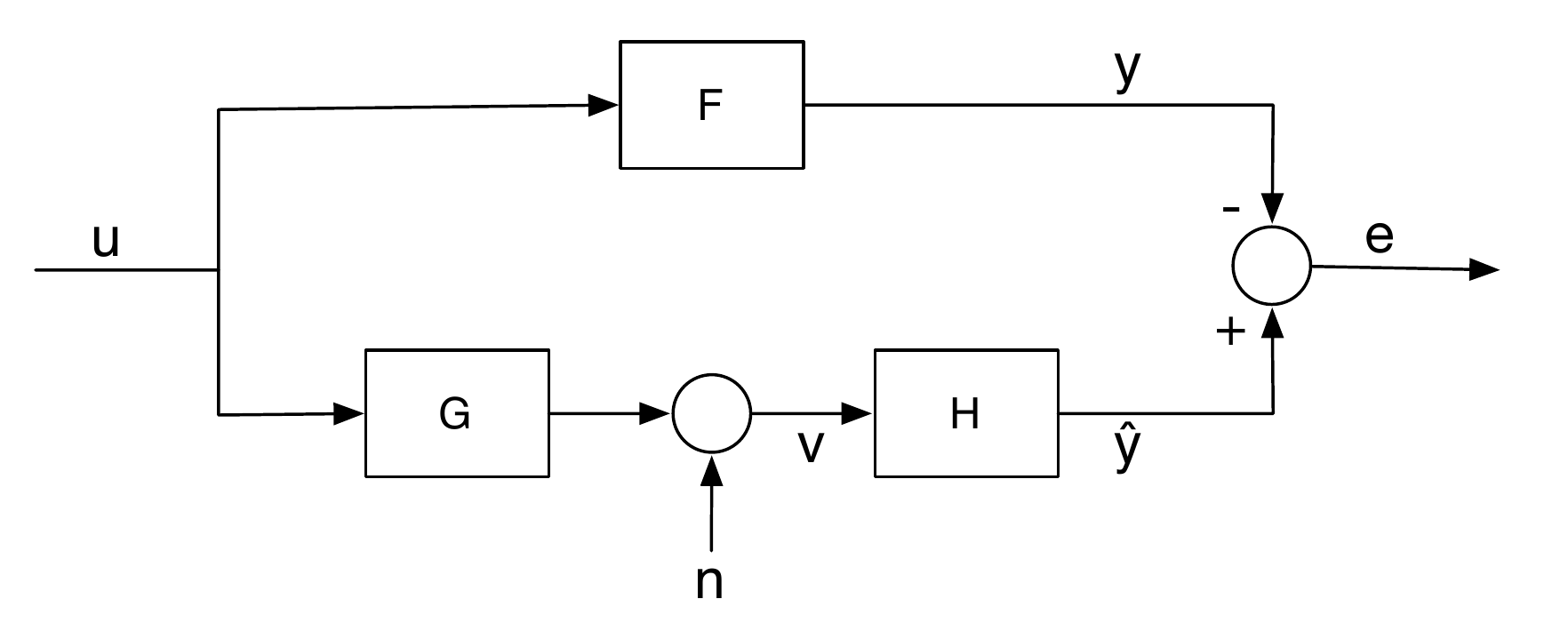}
\caption{Differentially private filter approximation set-up. For $v$ to be differentially private, we take $n$ to be a white Gaussian noise
with variance $\mathbb E[|n_t|^2] = d^2 \kappa_{\delta,\epsilon}^2 \|G\|_2^2$.}
\label{fig: DP filter approximation}
\end{figure}


\section{Linear Equalization Mechanisms}		\label{eq: linear equalization mechanisms}

\subsection{Linear Zero-Forcing Mechanism}

We first recall a mechanism initially described in \cite{LeNy_DP_CDC12}, which we call here the Linear 
Zero-Forcing (LZF) mechanism.
Note that once the differentially private signal $v=Gu+n$ is obtained, the task of estimating $y$ from $v$ is a standard
estimation (or equalization) problem. The LZF mechanism is based on the linear zero-forcing equalization idea, 
and its main advantage is that it requires no statistical information about the input signal $u$.
Let $G$ be a stable, minimum phase filter (hence invertible). Let $H = F G^{-1}$. 
To guarantee $(\epsilon,\delta)$-differential privacy, the noise $n$ is chosen to be white Gaussian 
with variance $d^2 \kappa_{\delta,\epsilon}^2 \|G\|^2_2$. 
The MSE for the LZF mechanism is then
\begin{align*}
\xi^{LZF} = d^2 \kappa^2_{\epsilon, \delta} \|G\|_2^2 \|FG^{-1}\|^2_2.
\end{align*}
The best possible choice of filters $G$ is then described in the following theorem \cite{LeNy_DP_CDC12}.

\vspace{0.2cm}
\begin{thm}	\label{thm: error for ZFE mechanism}
We have, for any stable, minimum phase system $G$, 
\[
\xi^{LZF} \geq d^2 \kappa_{\epsilon, \delta}^2 \left(\frac{1}{2 \pi} \int_{-\pi}^\pi |F(e^{j\omega})| d \omega \right)^2.
\]
This lower bound on the mean-squared error of the LZF mechanism is attained by letting 
$|G(e^{j \omega})|^2 = \lambda |F(e^{j \omega})|$ for all $\omega \in [-\pi,\pi)$,
where $\lambda$ is some arbitrary positive number. It can be approached arbitrarily closely by stable, 
rational, minimum phase transfer functions $G$. 
\end{thm}
\vspace{0.5cm}

Note that if $|F(e^{j \omega})|$ satisfies the Paley-Wiener condition
\[
\frac{1}{2 \pi} \int_{-\pi}^\pi \log |F(e^{j \omega})| d \omega > - \infty,
\]
then it has a spectral factorization $|F(e^{j \omega})| = \phi^+(\omega) \phi^-(\omega)$ and the bound of Theorem \ref{thm: error for ZFE mechanism}
is attained by taking $G$ with impulse reponse
\[
g_k = \frac{1}{2\pi} \int_{-\pi}^\pi \phi^+(\omega) e^{j \omega k} d \omega, \;\; k \geq 0.
\]
Note also that the MSE obtained for the best LZF mechanism in Theorem \ref{thm: error for ZFE mechanism}  is independent of the input signal $u$.
The design of $H$ does not attempt to minimize the effect of the noise $n$, as is the case with zero-forcing equalizers \cite{Proakis00_digitalComBook}.
The next section discusses another scheme that achieves a smaller error but requires some additional public knowledge about the statistics of the
input signal $u$.

\subsection{LMMSE Mechanism}		\label{section: MMSE mechanism}

The main issue with linear zero-forcing equalizers in communication systems is the noise amplification behavior 
at frequencies where $|G(e^{j \omega})|$ is small, due to the inversion in $H = F G^{-1}$. However, this issue is not as 
problematic for the optimal LZF mechanism, since in this case we essentially have $|H(e^{j \omega})|=\sqrt{|F(e^{j \omega})|}$, 
i.e., the amplification is compensated by the fact that $|F(e^{j \omega})|$ and $|G(e^{j \omega})|$ are both small 
at the same frequencies. Nonetheless, in this section we explore a scheme based on minimum
mean square equalization, which we call the Linear Minimum Mean Square Error (LMMSE) mechanism, and which 
can exhibit better performance than the LZF mechanism but requires some additional knowledge about the second order statistics of $u$.
This scheme was briefly discussed in \cite{LeNy_DP_CDC12}, but the optimization of $G$ described below was not
performed in that paper.

Hence, assume that that it is publicly known that $u$ is wide-sense stationary 
with know mean $\mu$ and autocorrelation $r_u[k] = \mathbb E[u_t u_{t-k}], \forall k$. Without loss of generality, we can then assume $\mu$ to
be zero, by substracting the known mean of $y$ equal to $F(1) \mu$.
The power spectral density of $u$ is denoted $P_u$, and is assumed to be rational for simplicity. 


The LMMSE mechanism is based on designing the filter $H$ as a Wiener filter in order to estimate $y$ from $v$. 
For tractability reasons, we derive the performance of the non-causal infinite impulse response Wiener filter, and optimize the choice of
$G$ with respect to this choice for $H$. Once $G$ is fixed, real-time consideration issues can force us to use a suboptimal design
with $H$ a causal Wiener filter, or perhaps introducing a small delay.

The non-causal Wiener filter $H$ has the transfer function
\[
H(z) = \frac{P_{yv}(z)}{P_v(z)},
\]
where $P_{yv}$ is the cross power spectral density of $y$ and $v$.
Since $w$ and $u$ are uncorrelated, we have
\[
P_{yv}(z) = P_u(z) F(z) G(z^{-1}).
\]
As for $P_v$, we have, with $n$ a white noise of variance $\sigma^2 = d^2 \kappa_{\delta,\epsilon}^2 \|G\|_2^2$,
\[
P_v(z) = P_u(z) G(z) G(z^{-1}) + \sigma^2.
\]
Hence 
\begin{align}		\label{eq: LMMSE filter}
H(z) = \frac{P_u(z) F(z) G(z^{-1})}{P_u(z) G(z) G(z^{-1}) + \kappa_{\delta,\epsilon}^2 \|G\|_2^2}.
\end{align}
The MSE can then be expressed as
\begin{align}
\xi^{LMMSE} &= \frac{1}{2 \pi} \int_{-\pi}^\pi \frac{P_u(e^{j \omega}) |F(e^{j \omega})|^2} 
{ \frac{P_u(e^{j \omega})}{d^2 \kappa_{\delta,\epsilon}^2} \frac{|G(e^{j \omega})|^2}{\|G\|_2^2} + 1 } d \omega. \label{eq: error IIR Wiener noncausal}
\end{align}

Note that we recover the LZF mechanism in the limit  $P_u(e^{j \omega}) >> d^2 \kappa(\delta,\epsilon)^2$.

\subsubsection{Privacy-Preserving Filter Optimization}		\label{section: MMSE optimization}

A close-to-optimal filter $G$ for the LMMSE mechanism can then be obtained by optimization, assuming initially that the reconstruction is done with the 
non-causal Wiener filter $H$. We discretize (\ref{eq: error IIR Wiener noncausal}) at the set of frequencies $\omega_i = \frac{i \pi}{N}, i=0 \ldots N$.
Note that all functions in the integral (\ref{eq: error IIR Wiener noncausal}) are even functions of $\omega$, hence we can restrict out attention to the interval
$[0,\pi]$.
We then define the $N+1$ variables 
\begin{align}		\label{eq: decision vars frequency domain}
x_i = \frac{|G(e^{j \omega_i})|^2}{\|G\|_2^2}, \;\; x_i \geq 0,
\end{align}
and the nonnegative constants 
\begin{align*}
\alpha_i &= P_u(e^{j \omega_i}) |F(e^{j \omega_i})|^2, \;\; i=0, \ldots N \\
\beta_i &= \frac{P_u(e^{j \omega_i})}{d^2 \kappa_{\delta,\epsilon}^2}, \;\;  i=0, \ldots N.
\end{align*}
The minimization of the error (\ref{eq: error IIR Wiener noncausal}) leads to the following problem 
(using a trapezoidal approximation of the integrals)
\begin{align}
\min_{\mathbf x} \;\;& \frac{1}{2N} \sum_{i=0}^{N-1} \frac{\alpha_i}{\beta_i x_i + 1} + \frac{\alpha_{i+1}}{\beta_{i+1} x_{i+1} + 1} \label{eq: optimization LMMSE-M} \\
\text{s.t. } \;\;& \frac{1}{2N} \sum_{i=0}^{N-1} x_i + x_{i+1} = 1 \label{eq: normalization constraint} \\
\;\;& x_i \geq 0, \;\; i=0, \ldots N. \nonumber
\end{align}
Note that the constraint (\ref{eq: normalization constraint}) comes from the fact that
\[
\frac{1}{\pi} \int_{0}^\pi \frac{|G(e^{j \omega})|^2}{\|G\|_2^2} d \omega
=\frac{1}{2\pi} \int_{-\pi}^\pi \frac{|G(e^{j \omega})|^2}{\|G\|_2^2} d \omega = 1.
\]

The optimization problem (\ref{eq: optimization LMMSE-M}) is convex, and can thus be solved efficiently even for fine discretizations of the interval $[0,\pi]$.
The transfer function of the filter $G$ can then be obtained for example by simple interpolation.

\begin{rem}
Even if the statistical assumptions on $u$ turn out not to be correct, the differential privacy guarantee of the LMMSE mechanism still holds and 
only its performance is impacted.
\end{rem}

\subsubsection{Causal Mechanism}		\label{section: causal MMSE}

The previous description of the LMMSE mechanism involves a possibly non-causal filter $H$. Sometimes, the anti-causal part of this filter might have 
a fast decreasing impulse response, in which case the scheme can be implemented approximately by introducing a small delay in the release 
of the output signal $\hat y$. Otherwise, we need to implement a causal Wiener filter $H$. Denoting the spectral factorization of $P_v$
\[
P_v(z) = \gamma_v^2 Q_v(z) Q_v(z^{-1}),
\]
we then have
\[
H(z) = \frac{1}{\gamma_v^2 Q_v(z)} \left[  \frac{P_{yv}(z)}{Q_v(z^{-1})}   \right]_+, 
\]
where, for a linear filter $L$ with impulse reponse $\{l_t\}_{- \infty \leq t \leq \infty}$, $[L(z)]_+$ denotes the causal filter with impulse response $\{l_t \mathbf 1_{\{t \geq 0\}}\}_t$.
Due to the more complex expression for $H$ and the resulting MSE, the design of the optimal filter $G$ in this case  is left for future work. Here, we optimize
$G$ assuming a possibly non-causal filter $H$, and then simply modify $H$ afterwards if causality needs to be enforced.


\section{Decision-Feedback Mechanisms}		\label{eq: DF mechanisms}

In general, solutions to the problem of reconstructing the optimum maximum-likelihood estimator of $\{(Fu)_k\}_{k \geq 0}$ from 
$\{v_k\}_{k \geq 0}$ are computationally intensive and require the knowledge of the full joint probability distribution of 
$\{u_k\}_{k \geq 0}$ \cite{Proakis00_digitalComBook}. This is the main reason why simpler linear architectures such as 
the one described in Section \ref{eq: linear equalization mechanisms} are more often implemented in communication receivers. 
However, so far, we have not exploited in the estimation procedures the knowledge that the input signal takes discrete values 
(or perhaps is even binary valued, as in \cite{Dwork10_DPcounter, Chan11_DPcontinuous}). This can be done by introducing only a slight
degree of nonlinearity, using the idea of decision-feedback equalization \cite{Proakis00_digitalComBook}. We call the resulting 
mechanism a Decision-Feedback (DF) mechanism. Its architecture is depicted on Fig. \ref{fig: DF mechanism}.

\begin{figure}
\includegraphics[width=\linewidth]{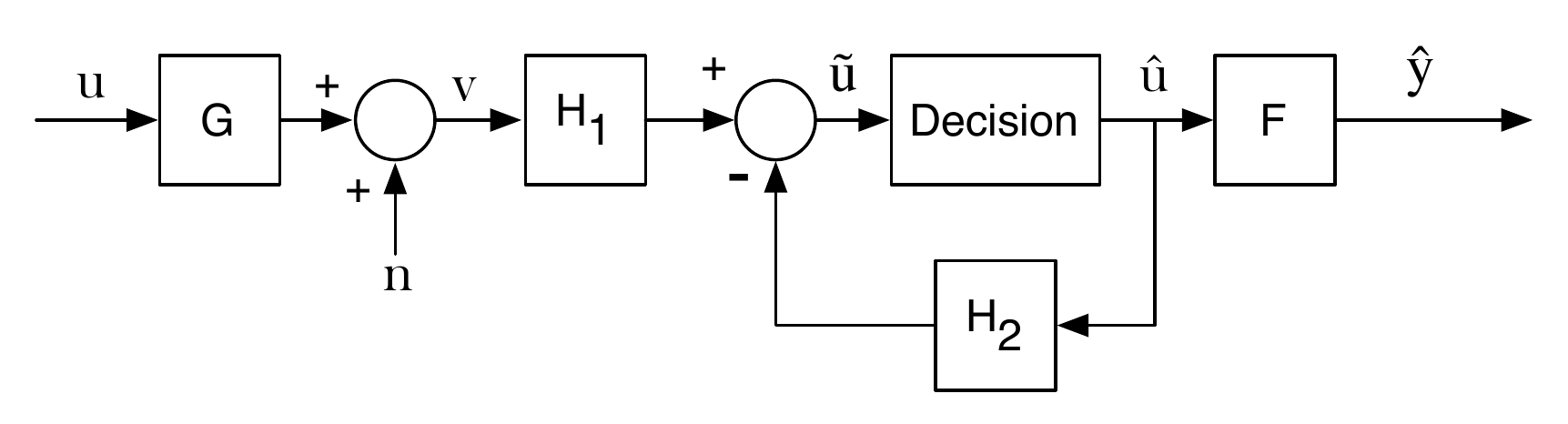}
\caption{Decision-feedback mechanism. The decision block is nonlinear and depends on the knowledge about the input signal $u$, 
acting as a detector/quantizer.}
\label{fig: DF mechanism}
\end{figure}

The second stage of a DF mechanism consists of a forward filter $H_1$, a nonlinear decision procedure (detector or quantizer) to estimate $u$ from $\tilde u$, 
which exploits the fact that $u$ takes discrete values, and a filter $H_2$ that feeds back the previous symbol decisions to correct the
intermediate estimate $\tilde u$. $H_2$ is assumed to be strictly causal, but generally $H_1$ is taken to be non-causal in standard equalizers,
for better performance \cite{Voois96_delayDFE}. Hence, DF mechanisms will typically introduce a small delay in the publication 
of the output signal $\hat y$. In the absence of detailed information about the distribution of $u$, the decision device can be a simple 
quantizer 
for integer valued input sequences, or a detector $\hat u_k = \texttt{sign}(\tilde u_k)$ for binary valued input sequences.

DF equalizers have a long history, and approximate expressions for their MSE can be derived \cite{Voois96_delayDFE}.
For tractability reasons, these derivations invariably make the simplifying assumption that the decisions $\hat u$ that enter the 
feedback filter are correct, i.e., $\hat u \equiv u$. Unfortunately, it appears that optimizing $G$ for the resulting approximate expression 
of the MSE is often not a good strategy, because the simplification results in a filter $G$ that does not need to be adapted to the
query $F$ any more (only to $P_u$). Still, we detail this optimization below and discuss an alternative design strategy for $G$ 
at the end of the section.

The error between the desired output $Fu$ and the signal $F \tilde u$, where $\tilde u$ is the input of the detector, is
\begin{align*}
e &= F(u - \tilde u) = F(u - H_1 v + H_2 \hat u),
\end{align*}
which, under the standard but simplifying assumption that $\hat u \equiv u$, gives
\[
e \approx F (B u - H_1 v), 
\]
with $B(z) = 1+H_2(z)$ a monic filter (since $H_2$ is strictly causal). As in section \ref{section: MMSE mechanism}, minimizing this approximate 
error (over possibly non-causal filters) requires $H_1$ to satisfy
\begin{align*}
H_1(z) &= B(z) \frac{P_{uv}(z)}{P_v(z)} \\
&= B(z) \frac{P_u(z) G(z^{-1})}{P_u(z) G(z) G(z^{-1}) + d^2 \kappa_{\delta,\epsilon}^2 \|G\|_2^2}.
\end{align*}
For this choice of $H_1$, the approximate MSE becomes
\begin{align}		\label{eq: DF approx error before B}
\xi^{DF} \approx \frac{1}{2 \pi} \int_{-\pi}^\pi \frac{P_u(e^{j \omega}) |B(e^{j \omega})|^2 |F(e^{j \omega})|^2} 
{ \frac{P_u(e^{j \omega})}{d^2 \kappa_{\delta,\epsilon}^2} \frac{|G(e^{j \omega})|^2}{\|G\|_2^2} + 1 } d \omega. 
\end{align}
Assuming now the spectral factorizations
\begin{align*}
P_u(e^{j \omega}) &= \gamma_u^2 |Q_u(e^{j \omega})|^2 \\
|F(e^{j \omega})|^2 &= \gamma_F^2 |Q_F(e^{j \omega})|^2 \\
\frac{P_u(e^{j \omega})}{d^2 \kappa_{\delta,\epsilon}^2} \frac{|G(e^{j \omega})|^2}{\|G\|_2^2} + 1 
&= \gamma^2 |Q(e^{j \omega})|^2,
\end{align*}
with $Q, Q_u$ and $Q_F$ canonical filters (monic, causal and minimum-phase), the approximate error (\ref{eq: DF approx error before B}) can be minimized by setting
\[
B(z) =  \frac{Q(z)}{Q_u(z) Q_F(z)}.
\]
The minimum approximate MSE is then
\begin{align}
&\xi^{DF} \approx \frac{\gamma_u^2 \gamma_F^2}{\gamma^2} \label{eq: final approximate MSE DFM} \\
&\approx \gamma_u^2 \gamma_F^2
\exp \left( {-\frac{1}{2 \pi} \int_{-\pi}^\pi \ln \left( \frac{P_u(e^{j \omega})}{d^2\kappa_{\delta,\epsilon}^2} \frac{|G(e^{j \omega})|^2}{\|G\|_2^2} + 1 \right) d \omega} \right).	\nonumber
\end{align}

The last expression is based on a well-known formula for $\gamma^2$, see \cite[p.105]{Hayes96_SPbook}.
Hence we see that an artifact of this approach is that the influence of $F$ and $G$ is decoupled,
and thus the minimization of (\ref{eq: final approximate MSE DFM}) over $G$ leads to a solution that is independent of $F$, 
which is generally undesirable. For example, for $u$ iid with $P_u(e^{j \omega}) \equiv 1$, optimizing (\ref{eq: final approximate MSE DFM}) 
gives the trivial solution $G(e^{j \omega}) \equiv 1$, and the whole mechanism reduces to an input perturbation scheme with
an additional decision stage. Nonetheless, for completeness we mention that optimizing 
(\ref{eq: final approximate MSE DFM}) over the choice of $G$ can be done using a discretization approach similar
to the one used in Section \ref{section: MMSE optimization}, now solving the convex optimization problem
\begin{align}
\max_{\mathbf x} \;\;& \frac{1}{2N} \sum_{i=0}^{N-1} \ln(\beta_i x_i + 1) + \ln(\beta_{i+1} x_{i+1} + 1) \label{eq: optimization DF-M} \\  
\text{s.t. } \;\;& \frac{1}{2N} \sum_{i=0}^{N-1} x_i + x_{i+1} = 1 \nonumber \\
\;\;& x_i \geq 0, \;\; i=0, \ldots N. \nonumber
\end{align}

In view of these issues, we mention an alternative design strategy for DF-mechanisms. Note from (\ref{eq: LMMSE filter}) that the (non-causal) 
LMMSE mechanism involves a reconstruction filter $H(z) = F(z) H_u(z)$, with $H_u$ the LMMSE estimator for $u$. Therefore we can interpret 
the DF mechanism on Fig. \ref{fig: DF mechanism} as introducing an additional stage to the linear mechanisms, to discretize the estimate of $u$,
and replacing $H_u$ by $H_1$. A strategy to improve on the performance of the LMMSE (or LZF) mechanism is then to keep the same 
prefilter $G$ designed in Section \ref{section: MMSE mechanism}, but simply replace the Wiener filter by a decision-feedback equalizer.
Our preliminary results tend to confirm that good performance is achievable with this strategy.

\section{Example}		\label{eq: example}

Consider approximating the filter
\[
F(z) = \frac{1+0.995 z^{-1}}{1-0.995 z^{-1}},
\]
with the privacy parameters set to $\epsilon = \ln 3$, $\delta = 0.05$.
The (wide-sense stationary) input signal is assumed to be binary valued, i.e., $u_t \in \left\{ \pm \frac{1}{2} \right\}$ for all $t$, with zero mean and power spectral
density 
\[
P_u(z) = \frac{3/4}{\left(1-\frac{1}{2}z^{-1}\right)\left(1-\frac{z}{2}\right)}.
\]
Such a signal can be generated by a two-state Markov chain in the stationary regime, with transition probability matrix
\[
\begin{bmatrix}
3/4 & 1/4 \\ 1/4 & 3/4
\end{bmatrix},
\]
one state corresponding to the input $-1/2$, and the other state corresponding to the input $1/2$, see, e.g., \cite{Brighenti09_MCspectrum}.
In this context we can imagine that the transitions are generated by individual users, and we want to prevent an adversary analyzing the trace 
$\{(Fu)_t\}_t$ to infer with confidence in which state the chain was at a particular time.

We designed four mechanism: LZF, LMMSE with $G$ optimized based on (\ref{eq: optimization LMMSE-M}), 
DF with $G$ optimized based on (\ref{eq: optimization DF-M}), and DF with the same $G$ as for the LMMSE mechanism.
The DF estimators introduce a $5$-period delay in the production of the estimate (finite impulse response equalizers were implemented here,
based on \cite{Voois96_delayDFE}).
Typical sample paths for these four mechanisms are shown on Fig. \ref{fig: sample paths}.
The theoretical root MSE (RMSE) for the LZF and (non-causal) LMMSE mechanisms are $8.82$ and $7.43$ respectively.
We see that the DF mechanisms significantly reduces the fluctuations in the produced output. 
Moreover, the LMMSE pre-filter $G$ leads to a clearly better performance for the DF mechanism than the one  
based on (\ref{eq: optimization DF-M}) in this case. The magnitude of the frequency response $|G(e^{j \omega})|$ is shown
on Fig. \ref{fig: Bode plots} for both filters. The cutoff of the LMMSE pre-filter occurs much earlier, taking into
account the fact that $F$ filters the high frequencies of $u$ anyway, and this helps to reduce the degradation
due to the privacy-preserving noise $n$.

\begin{figure}
\includegraphics[width=\linewidth]{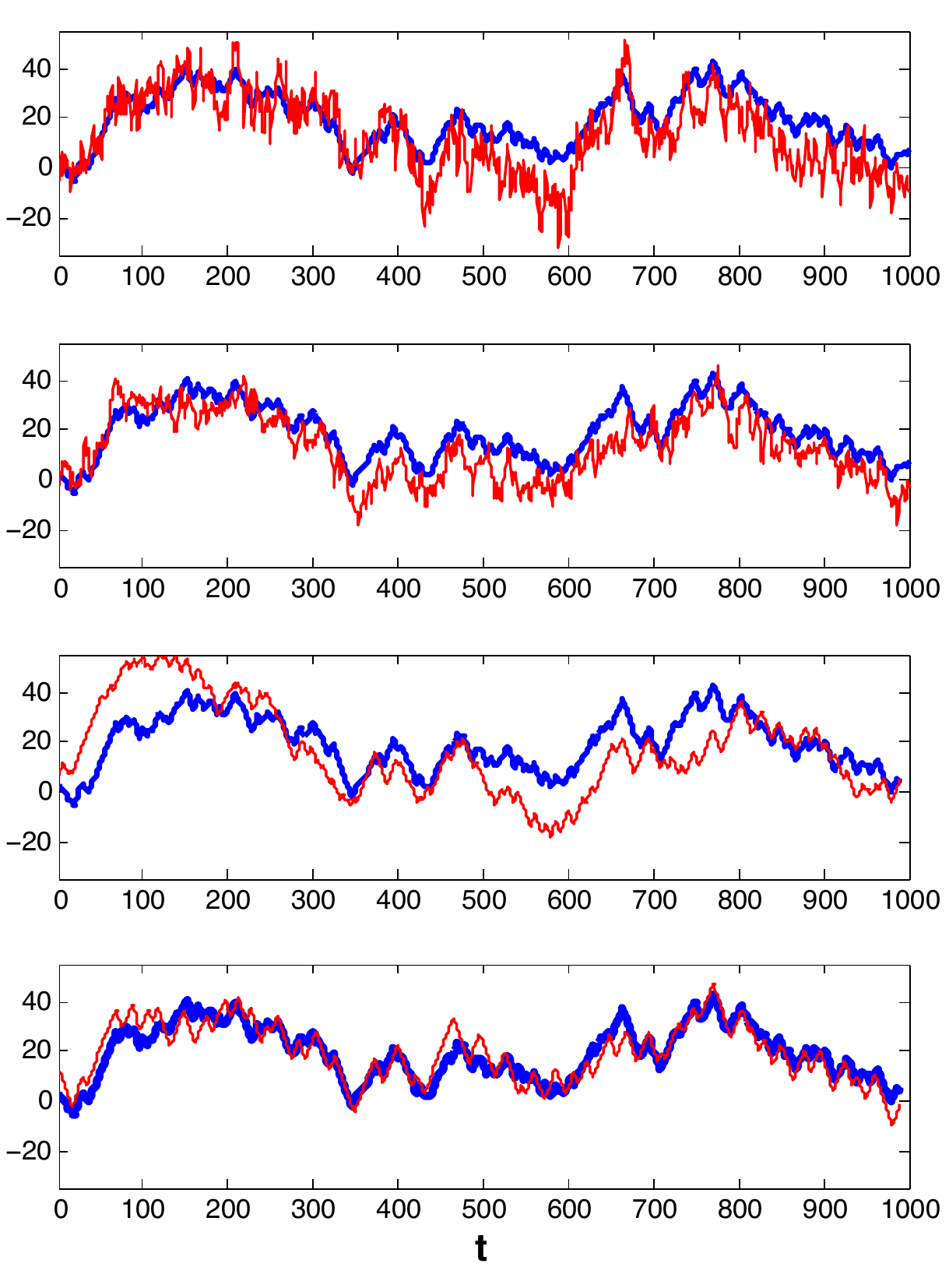}
\caption{Sample Paths for four mechanisms. From top to bottom: LZF, LMMSE, DF with $G$ optimized based on (\ref{eq: optimization DF-M}), 
and DF with the same $G$ as for the LMMSE mechanism. The original non-private output is shown as the thick blue line.} 
\label{fig: sample paths}
\end{figure}

\begin{figure}
\includegraphics[width=\linewidth]{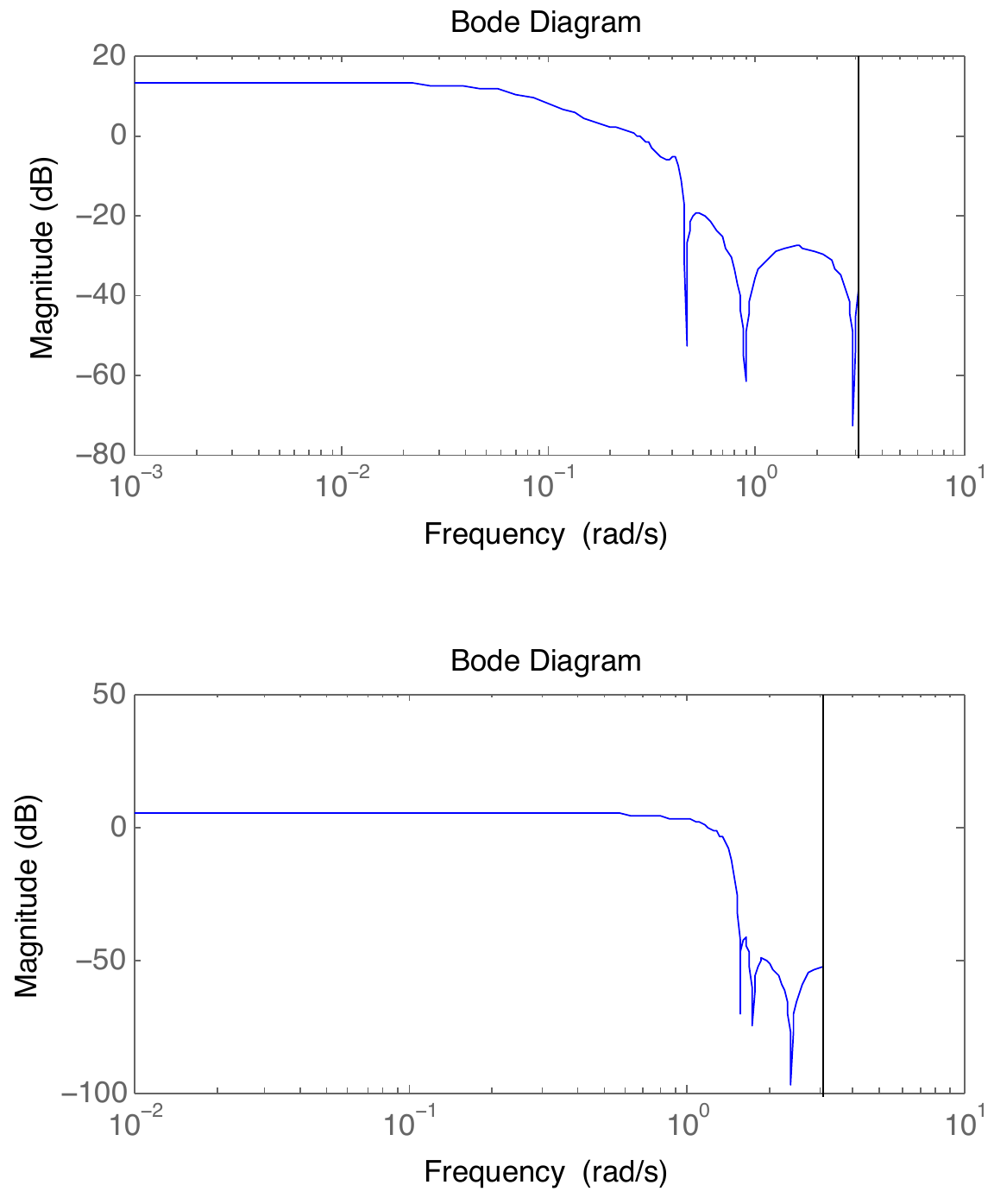}
\caption{Magnitude of the frequency reponse $|G(e^{j \omega})|$ for $G$ designed based on (\ref{eq: optimization LMMSE-M}) (top), 
and based on (\ref{eq: optimization DF-M}) (bottom).}
\label{fig: Bode plots}
\end{figure}






\section{Conclusions and Future Works}



In this paper, we have described several estimation techniques that can be leveraged
to minimize the impact on performance of a differential privacy specification for the filtering of event streams. 
The architecture considered here for the privacy mechanisms decomposes the problem into a standard 
equalization problem, for which many alternatives techniques could be used, and a first-stage privacy-preserving 
filter optimization problem. Future work on differentially private filtering for event streams includes enforcing privacy 
in scenarios where a single end-user can generate events at multiple times, optimizing SIMO and MIMO architectures
from a state-space perspective, and adaptive mechanisms that work in the absence of statistics for the input signals.

%



\bibliographystyle{IEEEtran}
\bibliography{IEEEabrv,/Users/jerome/Dropbox/Research/bibtex/energy,/Users/jerome/Dropbox/Research/bibtex/securityPrivacy,/Users/jerome/Dropbox/Research/bibtex/signalProcessing,/Users/jerome/Dropbox/Research/bibtex/controlSystems,/Users/jerome/Dropbox/Research/bibtex/communications} 



\end{document}